\documentclass[12pt]{article}
\usepackage{enumerate}
\usepackage{amsfonts}
\usepackage{amsmath}
\usepackage{amssymb}
\usepackage{amsthm}

\newcommand{\Tr}{\mathrm{Tr}}

\newcounter{defin}  \newcounter{lemma}  \newcounter{theorem}
\newcounter{property} \newcounter{corol}  \newcounter{remark} \newcounter{example}

\newenvironment{lemma}{\par\refstepcounter{lemma}%\noindent
     \textbf{Lemma \thelemma.} }{\rm\par}
\newenvironment{theorem}{\par\refstepcounter{theorem}%\noindent
     \textbf{Theorem \thetheorem.}\ }{\rm\par}
\newenvironment{property}{\par\refstepcounter{property}%\noindent
     \textbf{Proposition \theproperty.}\ }{\rm\par}
\newenvironment{corollary}{\par\refstepcounter{corol}%\noindent
     \textbf{Corollary \thecorol.} }{\rm\par}
\newenvironment{definition}{\par\refstepcounter{defin}%\noindent
     \textbf{Definition \thedefin.}\ }{\rm\par}
\newenvironment{remark}{\par\refstepcounter{remark}%\noindent
     \textbf{Remark \theremark.}}{\rm\par}
\newenvironment{example}{\par\refstepcounter{example}%\noindent
     \textbf{Example \theexample.}}{\rm\par}

\def\bra #1{\langle #1\vert}
\def\ket #1{\vert #1\rangle}

\def\ketbra #1#2{\vert #1\rangle \langle #2\vert}

\def\Id{\mathop{\rm Id}\nolimits}
\def\Ran{\mathop{\rm Ran}\nolimits}
\def\overco{\mathop{\overline{\rm co}}\nolimits}
\def\rank{\mathop{\rm rank}\nolimits}
\newcommand{\LOG}{\log}

\def\supp{\mathop{\rm supp}\nolimits}

\begin{document}
\title{Mutual and coherent informations for infinite-dimensional quantum channels\footnote{The work is
partially supported by RFBR grant 09-01-00424à and the scientific
program ``Mathematical control theory'' of RAS.}}
\author{A.S. Holevo, M.E. Shirokov\\
Steklov Mathematical Institute, RAS, Moscow\\holevo@mi.ras.ru,
msh@mi.ras.ru}
\date{}
\maketitle

\begin{abstract}
The work is devoted to study of quantum mutual information and
coherent information --  the two important characteristics of
quantum communication channel. Appropriate definitions of these
quantities in the infinite-dimensional case are given and their
properties are studied in detail. The basic identities relating
quantum mutual information and coherent information of a pair of
complementary channels are proved. An unexpected continuity property
of quantum mutual information and coherent information, following
from the above identities, is observed. The upper bound for the coherent
information is obtained.
\end{abstract}

\tableofcontents

\section{Introduction}

One of achievements in quantum information theory is discovery
of a whole number of important entropic and informational
characteristics of quantum systems (see e.g.\cite{QSCI,NC}). Some of
them such as the $\chi$-capacity and quantum mutual information have
direct classical analogs, others -- such as coherent information and
various entanglement measures -- do not have such analogs or they
are trivial.

Until recent time the main attention in quantum information theory
was paid to finite-dimensional systems, but recently considerable
interest to infinite-dimensional systems appeared: a broad class
important for applications in quantum optics constitute Bosonic
Gaussian systems \cite{E&W,QSCI}. Note that the properties of
the entropy and the relative entropy were studied in great detail,
including infinite-dimensional case, in connection with quantum
statistical mechanics, see e.g.\cite{L,O&P,W}. A study of entropic
and informational characteristics of quantum communication channels
from the general viewpoint of operator theory in separable Hilbert
space was undertaken in \cite{H-Sh-2,H-Sh}, where the quantities
related to the classical capacity, in the first place -- the
$\chi$-capacity, were investigated. The present work is devoted to
two other characteristics -- the quantum mutual information and the
coherent information. The first one is closely related to the
entanglement-assisted classical capacity while the second -- to the
quantum capacity of a channel. One of the author's goal was to give
an appropriate definition of these quantities in the
infinite-dimensional case, which would not require additional
artificial assumptions. The difficulty which was overcome is in the
uncertainties in expressions containing differences of entropies,
each of which can be infinite in the infinite-dimensional case. A
main result is Theorem 1 which implies that these quantities are
naturally defined and finite on the set of input states with finite
entropy, where they satisfy identities (\ref{basic1}) and
(\ref{basic2}) for complementary channels (the quantum mutual
information is defined uniquely for all input states but can be
infinite, still satisfying identity (\ref{basic1})).

In the introductory Section~2 a description of the corresponding
quantities for a finite quantum system is given. In Section~3 we
give the definition and study the properties of the quantum mutual
information in the infinite dimensional case. The main identity
(\ref{basic1}) for complementary channels is proved in Section~4.
Section~5 is devoted to the coherent information. In Section~6 we
point out somewhat unexpected continuity property of the mutual and
coherent informations implied by the identity (\ref{basic1}).

\section{Finite-dimensional case}

Consider the quantum system described by a finite-dimensional
Hilbert space $\mathcal{H}$ and denote by
$\mathfrak{S}(\mathcal{H})$ the convex set of \emph{quantum states}
described by density operators in $\mathcal{H}$, i.e. positive
operators with unit trace: $\rho\ge 0$, $\Tr\rho=1$. \emph{Entropy}
of the state $\rho$ (von Neumann entropy) is defined by the
relation\footnote{In the present paper $\log$ denotes the natural
logarithm.}
\begin{equation}\label{ent0}
H(\rho)=\Tr \eta (\rho),\quad \eta(x)=
\begin{cases}
-x \log x, & x>0,\\ 0,& x=0.
\end{cases}
\end{equation}

Let three systems $A$, $B$, $E$, described by the spaces
$\mathcal{H}_A$, $\mathcal{H}_B$, $\mathcal{H}_E$,
correspondingly, and an isometric operator $V\colon \mathcal{H}_A\to
\mathcal{H}_B\otimes \mathcal{H}_E$ be given, then the relations
\begin{equation}
\Phi(\rho)=\Tr_EV\rho V^*,\qquad \widetilde{\Phi}(\rho)=\Tr_BV\rho V^*,\quad \rho\in
\mathfrak{S}\left(\mathcal{H}_A\right),\label{compl}
\end{equation}
where $\Tr_X(\cdot)\doteq\Tr_{\mathcal{H}_X}(\cdot)$, define
completely positive trace-preserving maps, i.e. quantum channels
$\Phi\colon \mathfrak{S}\left(\mathcal{H}_A\right)\to
\mathfrak{S}\left( \mathcal{H}_B\right)$ and $\widetilde{\Phi}\colon
\mathfrak{S}\left( \mathcal{H}_A\right)\to \mathfrak{S}
\left(\mathcal{H}_E\right)$, which are called mutually
\emph{complementary} (this construction is generalized to the
infinite dimensional case without changes). The systems $A, B$
describe, correspondingly, input and output of the channel $\Phi$,
and $E$ -- its ``environment'' (see details in \cite{QSCI,NC}). The
identity operator in a space $\mathcal{H}_{X}$ and the identity
transformation of the set $\mathfrak{S}\left(
\mathcal{H}_{X}\right)$ will be denoted $I_X$ and $\Id_X$
correspondingly.

Let $\rho=\rho_A$ be an input state in the space  $\mathcal{H}_A$,
$\rho_B$ and $\rho_E$ be  the results of action of the channels
$\Phi$ and $\widetilde{\Phi}$ on the state $\rho_A$ correspondingly.
The \emph{quantum mutual information} is defined as follows
\begin{equation}\label{mutin1}
I(\rho,\Phi)= H(A)+H(B)-H(E),
\end{equation}
where the brief notations $H(A)=H(\rho_A)$, etc., are used \cite{A}.
By introducing the reference system $\mathcal{H}_R \cong
\mathcal{H}_A$ and the purification vector $\psi_{AR}\in
\mathcal{H}_{A}\otimes\mathcal{H}_{R}$ for the state $\rho_A$, the
mutual information can be represented as follows
\begin{equation}\label{mutin2}
I(\rho,\Phi)=H(R)+H(B)-H(BR),
\end{equation}
where $\rho_{BR}=\left(\Phi \otimes \Id_R \right)(\ketbra
{\psi_{AR}} {\psi_{AR}})$.

The mutual information $I(\rho,\Phi)$ have the several properties
similar to the properties of the Shannon information (see
Proposition \ref{concave} below). In \cite{E} (see also
\cite{QSCI}) it is shown that
\begin{equation}\label{cea}
\max_\rho I(\rho,\Phi)=C_{ea}(\Phi)
\end{equation}
is the classical entanglement-assisted capacity of the  channel
$\Phi$.

By introducing the analogous characteristic for the complementary
channel
\begin{equation}\label{mutin}
I(\rho,\widetilde{\Phi})=H(A)+H(E)-H(B)=H(R)+H(E)-H(ER),
\end{equation}
we have the following basic identity
\begin{equation}\label{main}
I(\rho,\Phi)+I(\rho,\widetilde{\Phi})=2H(\rho).
\end{equation}

An important component of the quantum mutual information $I(\rho,
\Phi)$ is the \emph{coherent information} (see \cite{BNS})
\begin{equation}\label{c-inf}
I_c(\rho,\Phi)=H(B)-H(E)=H(B)-H(RB).
\end{equation}
This notion is closely related to the quantum capacity of the
channel $\Phi$. In \cite{dev} (see also \cite{QSCI}) it is shown
that
\begin{equation}
Q(\Phi )=\lim_{n\to +\infty}\frac{1}{n}\max_\rho I_c(\rho,\Phi^{\otimes
n})\label{UBC1}
\end{equation}
is the quantum capacity of the channel $\Phi$. Identity (\ref{main})
is equivalent to the following one
\begin{equation}\label{main1}
I_c(\rho,\Phi)+I_c(\rho,\widetilde{\Phi})=0.
\end{equation}

The aim of this paper is to explore definitions and properties of
the analogs of the values  $I(\rho,\Phi)$ and $I_c(\rho,\Phi)$ in an
infinite dimensional Hilbert space. In particular, it will be shown
that the coherent information is naturally defined on the set of
states with finite entropy, where the analog of identity
(\ref{main1}) holds. The results of this paper can be used for
generalization of relations (\ref{cea}) and (\ref{UBC1}) to the case
of infinite dimensional channels.

\section{Mutual information}

In what follows $\mathcal{H}$ is a separable Hilbert space. Let
$\mathfrak{T}(\mathcal{H})$ be the Banach space of trace class
operators, so that $\mathfrak{S}(\mathcal{H}) \subset
\mathfrak{T}(\mathcal{H})$. Consider the natural extension of the
von Neumann entropy $H(\rho)=\Tr\eta(\rho)$ of a quantum state
$\rho\in\mathfrak{S}(\mathcal{H})$ to the cone
$\mathfrak{T}_{+}(\mathcal{H})$ of all positive  trace class
operators. \medskip

\begin{definition}\label{q-e}%\rm\cite{L}}
The \emph{entropy} of an operator $A \in
\mathfrak{T}_{+}(\mathcal{H})$ is defined as follows
\begin{equation}\label{ent}
H(A) =\Tr AH\left(\frac{A}{\Tr A}\right)=\Tr \eta (A)-\eta(\Tr A).
\end{equation}
\end{definition}\medskip

The entropy is a concave lower semicontinuous function on
the cone $\mathfrak{T}_{+}(\mathcal{H})$, taking values in $[0,
+\infty]$. By using Definition 1 and the well known properties of
the von Neumann entropy (see \cite{O&P}), it is easy to obtain the
following relations:
\begin{gather}
H(\lambda A)=\lambda H(A),\quad\lambda\ge 0,\label{H-fun-eq}\\ H(A)+H(B-A)\le H(B)\le
H(A)+H(B-A)+\Tr B\, h_2\left(\frac{\Tr A}{\Tr B}\right)\label{H-fun-ineq},
\end{gather}
where $A,B\in\mathfrak{T}_+(\mathcal{H})$, $A\le B$, and
$h_2(x)=\eta(x)+\eta(1-x)$.

We will also use the function $S(A)=\Tr\eta (A)$ on the cone
$\mathfrak{T}_{+}(\mathcal{H})$ coinciding with the function $H(A)$
on the set $\mathfrak{S}(\mathcal{H})$.\medskip

\begin{definition}\label{r-e}
The \emph{relative entropy} of operators
$A,B\in\mathfrak{T}_{+}(\mathcal{H})$ is defined as follows
$$
H(A\| B)=
\begin{cases}
\sum_{i=1}^{+\infty}\langle e_i|\,(A\log A-A\log B+B-A)\,|e_i\rangle, & \supp
A\subseteq \supp B,\\ +\infty, & \supp A\nsubseteq \supp B,
\end{cases}
$$
where $\left\{|e_i\rangle\right\}_{i=1}^{+\infty}$ is the
orthonormal basis of eigenvalues  of the operator $A$ \cite{L}.
\end{definition}\medskip

We will use the following lemma (\cite[Lemma~4]{L}).\medskip
\begin{lemma}
\emph{Let $\{ P_n\}$ be a nondecreasing sequence of projectors, converging
to the identity operator $I$ in the strong operator topology, and
$A,B$ be arbitrary positive trace class operators. Then the
sequences $\{H(P_nAP_n)\}$ and $\{ H(P_nAP_n\| P_nBP_n)\}$ are
nondecreasing,}
$$
H(A)=\lim_{n\to +\infty} H(P_nAP_n)\qquad\textit{and\/}\qquad H(A\| B)=\lim_{n\to
\infty} H(P_nAP_n\| P_nBP_n).
$$
\end{lemma}

\begin{definition}\label{channel} A \textit{quantum channel}
is a linear trace-preserving map\break
$\Phi:\mathfrak{T}(\mathcal{H}_A)\rightarrow
\mathfrak{T}(\mathcal{H}_B)\,$ such that the dual map $\Phi^{*}$ from
the $C^{*}$-algebra $\mathfrak{B}(\mathcal{H}_B)$ into the
$C^{*}$-algebra $\mathfrak{B}(\mathcal{H}_A)$ is completely positive
\cite{d}.
\end{definition}\medskip

In what follows we will use the fundamental \emph{monotonicity}
property of the relative entropy established in \cite{L} and
expressed by the inequality:
\begin{equation}\label{monoton}
H(\Phi(A)\| \Phi(B))\le H(A\| B)
\end{equation}
valid for an arbitrary  quantum channel $\Phi$ and arbitrary
positive trace class operators $A$ and $B$.\medskip

\begin{definition}\label{m-inf}
Let $\Phi\colon \mathfrak{S}(\mathcal{H}_A)\to
\mathfrak{S}(\mathcal{H}_B)$ be a quantum channel and $\rho$ be an
arbitrary quantum state in $\mathfrak{S}(\mathcal{H}_A)$ with the
spectral representation  $\rho=\sum_{i=1}^{+\infty} \lambda_i
\ket {e_i} \bra {e_i}$. The \emph{mutual information} of the channel
$\Phi$ at the state $\rho$ is defined as follows
$$
I(\rho, \Phi)=H(\Phi \otimes \Id_R (\ket {\varphi_\rho} \bra {\varphi_\rho}) \|
\Phi (\rho) \otimes \rho),
$$
where
\begin{equation}\label{purif}
\ket {\varphi_\rho}=\sum_{i=1}^{+\infty} \sqrt{\lambda_i} \ket {e_i} \otimes \ket
{e_i}\in \mathcal{H}_A \otimes \mathcal{H}_R
\end{equation}
is a purification vector\footnote{This means that $\Tr_R\ket
{\varphi_\rho}\bra {\varphi_\rho}=\rho$.} for the state $\rho$.
\end{definition}\medskip

Note that in the case  $\dim\mathcal{H}_A < +\infty$ and
$\dim\mathcal{H}_B < +\infty$ this definition is equivalent to
(\ref{mutin1}), (\ref{mutin2}), since
$$
\begin{array}{c}
H(\Phi \otimes \Id_R (\ket {\varphi_\rho}\bra {\varphi_\rho})\| \Phi (\rho)
\otimes \rho)=H(\rho_{BR}\| \rho_B \otimes\rho_R)\\\\ = \Tr (\rho_{BR}
(\LOG(\rho_{BR})-\LOG(\rho_B\otimes \rho_R)))\\\\=-H(\rho_{BR})+H(\rho_B)+H(\rho_R)=-H(BR)+H(B)+H(R).
\end{array}
$$

\begin{remark}\label{m-inf-r}
The above definition of the value $I(\rho, \Phi)$ does not depend on
the choice of the space $\mathcal{H}_R$ and of the purification
vector $\varphi_{\rho}$. This can be shown by using the well known
relation between different purification vectors of a given state
(see \cite{QSCI,NC}) and properties of the relative entropy. $\square$
\end{remark}\medskip

In the finite dimensional case concavity of the mutual information
as a function of $\rho$ on the set  $\mathfrak{S}(\mathcal{H}_A)$
follows from concavity of the conditional entropy
$H(EB)-H(E)$ \cite{A, QSCI}. In the case
$\dim\mathcal{H}_A =+\infty$ and $\dim\mathcal{H}_B <+\infty$ this
implies concavity of the mutual information as a function of $\rho$
on the set $\mathfrak{S}_\mathrm{f}(\mathcal{H}_A)=\{\rho \in
\mathfrak{S}(\mathcal{H}_A)\,|\,\mathrm{rank}\rho < +\infty\}$.

In what follows convergence of quantum states means convergence of
the corresponding density operators to a limit operator in the trace
norm, which is equivalent to the weak operator convergence (see
\cite{d} or \cite[Appendix A]{H-Sh-2}). Note that the entropy and
the relative entropy are lower semicontinuous in their arguments
with respect to this convergence \cite{W}.

Let $\mathfrak{F}(A,B)$ be the set of all quantum channels from
$\mathfrak{S}(\mathcal{H}_A)$ to $\mathfrak{S}(\mathcal{H}_B)$
endowed with the \emph{strong convergence} topology \cite{H-Sh}.
Strong convergence of a  sequence
$\{\Phi_n\}\subset\mathfrak{F}(A,B)$ to a quantum
 channel $\Phi_0\in\mathfrak{F}(A,B)$ means that
$\lim_{n\to+\infty}\Phi_n(\rho)=\Phi_0(\rho)$ for any state
$\rho\in\mathfrak{S}(\mathcal{H}_A)$.

The following proposition is devoted to
generalization of the observations in \cite{A} to the infinite
dimensional case.\medskip

\begin{property}\label{concave} \emph{The function $(\rho, \Phi)\mapsto I(\rho, \Phi)$ is
nonnegative and  lower semicontinuous on the set
$\mathfrak{S}(\mathcal{H}_A)\times\mathfrak{F}(A,B)$. It has the
following properties:}
\begin{enumerate}[1)]
\item\vskip-5pt
concavity in $\rho$:\emph{ $I(\lambda\rho_{1}+(1-\lambda)\rho_2,
\Phi)\ge \lambda I(\rho_1, \Phi)+(1-\lambda)I(\rho_2, \Phi)$\textup;}
\item
convexity in $\Phi$: \emph{$I(\rho,
\lambda\Phi_1+(1-\lambda)\Phi_2)\le \lambda I(\rho,
\Phi_1)+(1-\lambda)I(\rho, \Phi_2)$\textup;}
\item
the 1-th chain rule: \emph{for arbitrary channels $\Phi\colon
\mathfrak{S}(\mathcal{H}_A)\to \mathfrak{S}(\mathcal{H}_B)$ and
$\Psi\colon \mathfrak{S}(\mathcal{H}_B)\to
\mathfrak{S}(\mathcal{H}_C)$ the inequality $\,I(\rho,
\Psi\circ\Phi)\le I(\rho, \Phi)\,$ holds for any
$\,\rho\in\mathfrak{S}(\mathcal{H}_A)$\textup;}
\item
the 2-th chain rule: \emph{for arbitrary channels $\Phi\colon
\mathfrak{S}(\mathcal{H}_A)\to \mathfrak{S}(\mathcal{H}_B)$ and
$\Psi\colon \mathfrak{S}(\mathcal{H}_B)\to
\mathfrak{S}(\mathcal{H}_C)$ the inequality $\,I(\rho,
\Psi\circ\Phi)\le I(\Phi(\rho), \Psi)\,$ holds for any
$\,\rho\in\mathfrak{S}(\mathcal{H}_A)$\textup;}
\item
subadditivity: \emph{for arbitrary channels $\Phi\colon
\mathfrak{S}(\mathcal{H}_A) \to \mathfrak{S}(\mathcal{H}_B)$ and \\
$\Psi\colon \mathfrak{S}(\mathcal{H}_C)\to
\mathfrak{S}(\mathcal{H}_D)$ the inequality
\begin{equation}\label{subaddit}
I(\omega, \Phi\otimes\Psi)\le I(\omega_A,\Phi)+I(\omega_{C},\Psi)
\end{equation}
holds for any
$\omega\in\mathfrak{S}(\mathcal{H}_A\otimes\mathcal{H}_C)$.}
\end{enumerate}
\end{property}

\begin{proof} Nonnegativity of the value $I(\rho, \Phi)$
follows from nonnegativity of the relative entropy. By Lemma
\ref{purification} below and Remark \ref{m-inf-r} lower
semicontinuity of the function $(\rho, \Phi)\mapsto I(\rho, \Phi)$
follows from lower semicontinuity of the relative entropy in the
both arguments.

To prove concavity of the function $\rho\mapsto I(\rho, \Phi)$
suppose first that  $\dim\mathcal{H}_B$ is finite. Let $\rho =
\alpha \sigma_1 + (1-\alpha) \sigma_2$ and  $\{ P_n \}$ be an
increasing sequence of finite rank spectral projectors of the state
$\rho$ strongly converging to $I_{A}$. Let
$$ \rho_n=\frac {P_n \rho
P_n} {\Tr P_n \rho}=\frac{\alpha P_n \sigma_1 P_n+(1-\alpha) P_n
\sigma_2 P_n} {\alpha \Tr P_n \sigma_1 + (1-\alpha) \Tr P_n
\sigma_2}=\frac {\mu_1^n \sigma_1^n+\mu_2^n \sigma_2^n} {\mu_1^n +
\mu_2^n},
$$
where
$$
\begin{gathered}
\mu_1^n=\alpha \Tr P_n \sigma_1,\qquad \sigma_1^n=\alpha \frac {P_n \sigma_1
P_n}{\mu_1^n},\\ \mu_2^n=(1-\alpha) \Tr P_n \sigma_2,\qquad \sigma_2^n=(1-\alpha)
\frac {P_n \sigma_2 P_n} {\mu_2^n}.
\end{gathered}
$$
By  concavity of the function $\rho\mapsto I(\rho,\Phi)$ on the set
$\mathfrak{S}_\mathrm{f}(\mathcal{H}_A)$ mentioned before
Proposition \ref{concave} we have
$$
I(\rho_n,\Phi)\ge \frac {\mu_1^n} {\mu_1^n+\mu_2^n} I(\sigma_1^n, \Phi)+\frac
{\mu_2^n} {\mu_1^n+\mu_2^n} I(\sigma_2^n, \Phi).
$$
Lemma $\ref{i_n}$ below implies $\lim_{n\to +\infty} I(\rho_n,
\Phi)=I(\rho, \Phi)$. By using lower semicontinuity of the function
$\rho\mapsto I(\rho, \Phi)$, we obtain
$$
\begin{aligned}
I(\rho, \Phi)&\ge \liminf_{n\to +\infty} \frac {\mu_1^n} {\mu_1^n+\mu_2^n}
I(\sigma_1^n, \Phi)+\liminf_{n\to +\infty}\frac {\mu_2^n} {\mu_1^n+\mu_2^n}
I(\sigma_2^n, \Phi)\\ &\ge \alpha I(\sigma_1, \Phi) + (1-\alpha)I(\sigma_2,\Phi).
\end{aligned}
$$

Let $\Phi$ be an arbitrary quantum channel. Consider the sequence of
channels $\Phi_n=\Pi_n \circ \Phi$ with finite dimensional output,
where
$$
\Pi_n(\rho)=P_n \rho P_n+[\Tr ((I-P_n)\rho)]\ketbra \psi \psi
$$
is a quantum channel from $\mathfrak{S}(\mathcal{H}_B)$ to itself
for each $n$,  $\{P_n\}$ is an increasing sequence of finite rank
projectors strongly converging to $I_{B}$, $\ketbra \psi \psi$ is a
fixed pure state in $\mathfrak{S}(\mathcal{H}_B)$. Then for each $n$
the function $\rho\mapsto I(\rho, \Phi_n)$ is concave by the above
observation. Since
$$
I(\rho, \Phi_n)\le I(\rho, \Phi)\;\;\forall n\qquad \textrm{and}\qquad \liminf_{n\to
+\infty} I(\rho, \Phi_n)\ge I(\rho, \Phi)
$$
by monotonicity of the relative entropy and lower semicontinuity of
the function $\Phi\mapsto I(\rho, \Phi)$, we have
$$
I(\rho, \Phi)=\lim_{n\to
+\infty} I(\rho, \Phi_n),
$$
Hence the  function $\rho\mapsto I(\rho, \Phi)$ is concave as a
pointwise limit of a sequence  of concave functions.

Convexity of the function $\Phi\mapsto I(\rho, \Phi)$ follows from
joint convexity of the relative entropy in their arguments \cite{W}.

The 1-th chain rule immediately follows from Definition \ref{m-inf}
and monotonicity of the relative entropy.

The 2-th chain rule is also proved by using monotonicity of the
relative entropy as follows.

Let $\ketbra \varphi \varphi$ be a purification of the state
$\rho\in\mathfrak{S}(\mathcal{H}_A)$ in the space
$\mathcal{H}_A\otimes\mathcal{H}_R$, then $\ketbra \psi
\psi=V\otimes I_R\ketbra \varphi \varphi V^*\otimes I_R$ is a
purification of the state $\Phi(\rho)\in\mathfrak{S}(\mathcal{H}_B)$
in the space $\mathcal{H}_B\otimes\mathcal{H}_E\otimes\mathcal{H}_R$
($V$ is the isometry from representation (\ref{compl}) of the
channel $\Phi$). Hence
$$
I(\Phi(\rho),\Psi)=H\left(\Psi\otimes \Id_{ER}(\ketbra \psi \psi )\|
\,\Psi(\Tr_{ER}\ketbra \psi \psi )\otimes\Tr_B\ketbra \psi \psi\right).
$$
A direct verification shows that taking the partial trace over the
space $\mathcal{H}_E$ on each arguments of the relative entropy in
the above expression transforms the right side of this expressions
to
$$
H\left((\Psi\circ\Phi)\otimes \Id_R(\ketbra \varphi \varphi )\|
\,(\Psi\circ\Phi)(\Tr_R\ketbra \varphi \varphi )\otimes\Tr_A\ketbra \varphi \varphi
\right)=I(\rho,\Psi\circ\Phi).
$$

The subadditivity property of the mutual information will be derived
from the corresponding property of this characteristics for finite
dimensional channels \cite{A,QSCI}.

Let $\{Q^X_{n}\}$ be an increasing sequence of finite rank
projectors in the space $\mathcal{H}_{X}$, strongly converging to
the operator $I_{X}$, where  $X=B,D$. The sequence of  channels
$$
\Pi^X_n(\rho)=Q^X_n\rho Q^X_n+ \bigl(\Tr(I_X-Q^X_n)\rho\bigr)\tau_X
$$
from $\mathfrak{S}(\mathcal{H}_{X})$ to itself, where $\tau_X$ is an
arbitrary pure state in the space $\mathcal{H}_{X}$, strongly
converges to the channel $\Id_X$.

Let $\omega$ be an arbitrary state in
$\mathfrak{S}(\mathcal{H}_A\otimes\mathcal{H}_C)$. Let $\{P^X_{n}\}$
be an increasing sequence of  finite rank  spectral projectors of
the state $\omega_X$, strongly converging to the operator $I_{X}$,
where $X=A,C$.

Consider the sequence of states
$$
\omega^n= \bigl(\Tr\bigl((P^A_n\otimes
P^C_n)\cdot\omega\bigr)\bigr)^{-1}(P^A_n\otimes P^C_n)\cdot\omega\cdot (P^A_n\otimes
P^C_n),
$$
converging to the state $\omega$.

A direct verification shows that
$$
\lambda_n\omega^n_X\le\omega_X,\quad X=A,C,\quad \textup{where}\;
\lambda_n=\Tr\bigl((P^A_n\otimes P^C_n)\cdot\omega\bigr).
$$

By Lemma \ref{approximation} below we have
\begin{equation}\label{key-ineq}
\lim_{n\to+\infty}I(\omega^n_A,\Pi^B_n\circ\Phi)=I(\omega_A,\Phi)\;\;
\textrm{and}\;\,
\lim_{n\to+\infty}I(\omega^n_{C},\Pi^D_n\circ\Psi)=I(\omega_{C},\Psi).
\end{equation}

Subadditivity of the mutual information for finite dimensional
channels implies
$$
I(\omega^n, (\Pi^B_n\circ\Phi)\otimes(\Pi^D_n\circ\Psi))\le
I(\omega^n_A,\Pi^B_n\circ\Phi)+I(\omega^n_{C},\Pi^D_n\circ\Psi).
$$
By (\ref{key-ineq}) and lower semicontinuity of the mutual
information as a function of a pair (state, channel) passing to the
limit in this inequality implies (\ref{subaddit}).
\end{proof}

In the proof of Proposition \ref{concave} the following lemmas were
used.\medskip
\begin{lemma}\label{purification}
\emph{Let $\,\mathcal{H}$ be a separable Hilbert space. For an arbitrary
sequence $\{\rho_n\}\subset\mathfrak{S}(\mathcal{H})$, converging to
a state $\rho_0$, there exists a corresponding purification sequence
$\{\hat{\rho}_n\}\subset\mathfrak{S}(\mathcal{H}\otimes\mathcal{H})$,
converging to a purification $\hat{\rho}_0$ of the state $\rho_0$.}
\end{lemma}

\begin{proof} The assertion of the lemma follows from the inequality
$$
\beta (\rho, \sigma)^2\le \|\rho-\sigma\|_1
$$
for the Bures distance  $\beta (\rho, \sigma) = \inf \|
\varphi_{\rho} - \varphi_{\sigma} \|$, where the infimum is over all
purification vectors $\varphi_{\rho}$ and $\varphi_{\sigma}$ of the
states $\rho$ and $\sigma$ \cite{QSCI,NC}.
\end{proof}

\begin{lemma}\label{i_n} \emph{Let $\Phi\colon \mathfrak{S}(\mathcal{H}_A)\to
\mathfrak{S}(\mathcal{H}_B)$ be a quantum channel such that
$\dim\mathcal{H}_B<+\infty$ and  $\rho_0$ be a state in
$\mathfrak{S}(\mathcal{H}_A)$ with the spectral representation
$\rho_0=\sum_{i=1}^{+\infty} \lambda_i \ket {e_i} \bra {e_i}$.
Let
\begin{equation}\label{spect}
\rho_n=\frac {1}{\mu_n}\sum_{i=1}^n \lambda_i \ket {e_i} \bra {e_i},\quad
\text{where\/}\quad \mu_n=\sum_{i=1}^n \lambda_i,
\end{equation}
for every $n$. Then $\,\lim_{n\to +\infty} I(\rho_n,
\Phi)=I(\rho_0, \Phi)$.}
\end{lemma}

\begin{proof}
Let $P_n=\sum_{i=1}^n \ket {e_i} \bra {e_i}$,
$n=1,2,\ldots\strut$ Since $\dim \mathcal{H}_B <\infty$, the
following value is finite
$$
\begin{aligned}
I_n&=H(\Phi \otimes \Id_R(\hat {\rho}_n)\| \Phi(\rho_0) \otimes \rho_n)\\
&=\mu_n^{-1} H(Q_n\left(\Phi \otimes \Id_R(\hat {\rho}_0) \right) Q_n\| Q_n
\left(\Phi(\rho_0) \otimes \rho_0 \right)Q_n),
\end{aligned}
$$
where
$$
\hat{\rho}_0=\sum_{i,j=1}^{+\infty} \sqrt{\lambda_i \lambda_j} \ket {e_i} \bra {e_j}
\otimes \ket {e_i} \bra {e_j},\qquad \hat{\rho}_n=\mu_n^{-1}\sum_{i,j=1}^n
\sqrt{\lambda_i \lambda_j} \ket {e_i} \bra {e_j} \otimes \ket {e_i} \bra {e_j}
$$
and $Q_n=I_B \otimes P_n$. By Lemma 1 we have
\begin{equation}\label{l-rel}
\lim_{n\to +\infty} I_n=H(\Phi \otimes \Id_R (\hat{\rho}_0) \| \Phi(\rho_0) \otimes
\rho_0)=I(\rho_0, \Phi)\le+\infty.
\end{equation}

We will prove that $\lim_{n\to +\infty} I_n=\lim_{n\to
+\infty} I(\rho_n, \Phi)$ by considering the difference
$I_n-I(\rho_n, \Phi)$. Since $H(\Phi \otimes
\Id_R(\hat{\rho}_n))<+\infty$, we have
\begin{multline*}
\begin{aligned}
I_n-I(\rho_n, \Phi)&=H(\Phi \otimes \Id_R (\hat{\rho}_n)\| \Phi(\rho_0) \otimes
\rho_n)-H(\Phi \otimes \Id_R (\hat{\rho}_n)\| \Phi({\rho_n}) \otimes {\rho_n})\\
&=-H(\Phi \otimes \Id_R(\hat{\rho}_n))-\Tr(\Phi \otimes \Id_R (\hat{\rho}_n))(\LOG
\Phi(\rho_0)\otimes \rho_n)\\
\end{aligned}\\
+H(\Phi \otimes \Id_R(\hat{\rho}_n))+\Tr(\Phi \otimes \Id_R (\hat{\rho}_n))\LOG
(\Phi(\rho_n)\otimes \rho_n)=A-B,
\end{multline*}
where
$$
\begin{aligned}
A &=-\Tr(\Phi \otimes \Id_R (\hat{\rho}_n))\LOG(\Phi(\rho_0)\otimes \rho_n),\\
B&=-\Tr(\Phi \otimes \Id_R (\hat{\rho}_n))\LOG(\Phi(\rho_n)\otimes \rho_n).
\end{aligned}
$$

We will use the following property of logarithm
\begin{equation}\label{log}
\LOG(\rho \otimes \sigma)=\LOG(\rho) \otimes I+I \otimes \LOG(\sigma),
\end{equation}
where in the case of not-full-rank states $\rho$ and $\sigma$ the
restrictions to the subspaces $\supp (\rho)$ and $\supp (\sigma)$
are kept in mind, that is
\begin{equation}
P_\rho\otimes P_{\sigma}\,(\LOG(\rho \otimes \sigma))= (P_\rho \LOG(\rho) P_\rho)
\otimes P_{\sigma}+P_\rho\otimes (P_{\sigma}\LOG(\sigma) P_{\sigma}),
\end{equation}
where $P_\rho$ and $P_{\sigma}$ are respectively the projectors onto
$\supp (\rho)$ and $\supp (\sigma)$. Since $P_{\Phi(\rho_n)}\le
P_{\Phi(\rho_0)}$, we have
$$
\begin{aligned}
A &=-\Tr(\Phi \otimes \Id_R(\hat{\rho}_n))(\LOG\Phi(\rho_0)\otimes I_R)-\Tr(\Phi
\otimes \Id_R(\hat{\rho}_n))(I_B \otimes \LOG(\rho_n))\\ &=-\Tr\Phi
(\rho_n)\LOG\Phi(\rho_0)+H(\rho_n).
\end{aligned}
$$
In the similar way we obtain
$$
\begin{aligned}
B&=-\Tr(\Phi \otimes \Id_R (\hat{\rho}_n))(\LOG\Phi(\rho_n)\otimes I_R)-\Tr(\Phi
\otimes \Id_R (\hat{\rho}_n))(I_B \otimes \LOG (\rho_n))\\
&=H(\Phi(\rho_n))+H(\rho_n).
\end{aligned}
$$
Hence,
$$
I_n-I(\rho_n, \Phi)=A-B =-\Tr~\Phi (\rho_n)\LOG\Phi(\rho_0)-H(\Phi(\rho_n))=H(\Phi
(\rho_n)\| \Phi(\rho_0)).
$$
By monotonicity of the relative entropy we have
$$
H(\Phi (\rho_n)\| \Phi(\rho_0))\le H(\rho_n\| \rho_0)=- \sum_{i=1}^n \frac
{\lambda_i} {\mu_n}\log\mu_n=- \log\mu_n.
$$
Since $\mu_n\to 1$ as $n\to +\infty$,  $\lim_{n\to
+\infty}(I_n-I(\rho_n,\Phi))=0$. This and (\ref{l-rel}) imply
$$
\lim_{n\to +\infty} I(\rho_n,\Phi)=I(\rho_0,\Phi)\le+\infty.
$$
\end{proof}

\begin{lemma}\label{approximation}
\emph{Let  $\Phi$ be an arbitrary channel from
$\,\mathfrak{S}(\mathcal{H}_A)$ to $\,\mathfrak{S}(\mathcal{H}_B)$
and $\{\Pi_{n}\}$ be a sequence of channels from
$\,\mathfrak{S}(\mathcal{H}_B)$ to $\,\mathfrak{S}(\mathcal{H}_B)$,
strongly converging to the identity channel. Let $\{\rho_{n}\}$ be a
sequence of states in $\,\mathfrak{S}(\mathcal{H}_A)$ converging to
a state $\rho_{0}$ such that $\lambda_n\rho_n\le\rho_0$ for some
sequence $\{\lambda_{n}\}$ converging to $\,1$. Then}
$$
\lim_{n\to+\infty}I(\rho_n,\Pi_n\circ\Phi)=I(\rho_0,\Phi).
$$
\end{lemma}

\begin{proof}
 It follows from the inequality $\lambda_n\rho_n\le\rho_0$ that $\rho_0=\lambda_n\rho_n+(1-\lambda_n)\sigma_n$, where $\sigma_n$
is a state in $\mathfrak{S}(\mathcal{H}_A)$.  Hence concavity and
nonnegativity of the mutual information and the 1-th chain rule imply
the inequality
$$
\lambda_nI(\rho_n,\Pi_n\circ\Phi)\le I(\rho_0,\Pi_n\circ\Phi)\le I(\rho_0,\Phi),
$$
showing that
$\limsup_{n\to+\infty}I(\rho_n,\Pi_n\circ\Phi)\le
I(\rho_0,\Phi)$. This and lower semicontinuity of the function
$(\rho,\Phi)\mapsto I(\rho,\Phi)$ imply the assertion of the
lemma.
\end{proof}

\section{The relation between mutual informations of complementary channels}

The main result of this section is an infinite dimensional
generalization of relation (\ref{main}) between mutual informations
of a pair of complementary channels (nontriviality of this result is
connected with a possible uncertainty ``$\infty -\infty$'' in
expressions (\ref{mutin2}) and (\ref{mutin})).

Let $\mathcal{H}_A$, $\mathcal{H}_B$, $\mathcal{H}_E$ be separable
Hilbert spaces and $V\colon \mathcal{H}_A\to \mathcal{H}_B\otimes
\mathcal{H}_E$ be an isometry, then relations (\ref{compl}) define a
pair of complementary channels $\Phi, \widetilde{\Phi}$ similar to
the finite dimensional case.\medskip

\begin{theorem} \emph{For an arbitrary state $\rho\in\mathfrak{S}(\mathcal{H}_A)$
the following relation holds:}
\begin{equation}\label{basic1}
I(\rho,\Phi)+I(\rho,\widetilde{\Phi})=2H(\rho).
\end{equation}
\end{theorem}

\begin{proof} Let
$\left\{\ket {h_i}\right\}_{i=1}^{+\infty}$ be an orthonormal basis
in the space $\mathcal{H}_E$, then
$$
V\ket\varphi= \sum_{i=1}^{+\infty} V_i\ket\varphi\otimes \ket{h_i},
$$
where $\;V_i\colon \mathcal{H}_A\to\mathcal{H}_B\;$ is a sequence of
bounded operators, satisfying the condition
$\sum_{i=1}^{+\infty} V_i^* V_i=I_A$, the channel $\Phi$ has the
Kraus representation  $\Phi(\rho)=\sum_{i=1}^{+\infty} V_i \rho
V_i^*$, the complementary channel $\widetilde{\Phi}$ has the
representation $\widetilde{\Phi}(\rho)=\sum_{i,j=1}^{+\infty}
\left[\Tr V_i \rho V_j^*\right] \ketbra {h_i} {h_j}$
(cf.\cite{H-c-c}).

Let $\rho=\sum_{i=1}^m \lambda_i |e_i\rangle\langle e_i|$ be
a finite rank state in $\mathfrak{S}(\mathcal{H}_A)$ and
$\hat{\rho}$ be its purification in $\mathfrak{S}(\mathcal{H}_A
\otimes \mathcal{H}_R)$. Consider the sequence of quantum
operations\footnote{A {\it quantum operation} is a linear completely
positive trace non-increasing map \cite{QSCI, NC}.} $\Phi_n(\rho) =
\sum_{i=1}^n V_i \rho V_i^*$. The sequence $\{\Phi_n\}$
strongly and monotonously converges to the channel $\Phi$ (that is
$\Phi_n(\rho)\le\Phi_{n+1}(\rho)$ for all $n$ and
$\rho\in\mathfrak{S}(\mathcal{H}_A)$).

Since $\Phi_n \otimes \Id_R(\hat{\rho})$ is a finite rank state, we
have
$$
\begin{aligned}
X_n&=H\left(\Phi_n \otimes \Id_R(\hat{\rho})\| \Phi(\rho)\otimes \rho\right)\\
&=-S(\Phi_n\otimes \Id_R(\hat{\rho}))-\Tr(\Phi_n \otimes \Id_R(\hat{\rho}))\LOG
(\Phi(\rho)\otimes \rho)+R_n,
\end{aligned}
$$
where $R_n=1-\Tr\Phi_n(\rho)\to 0$ as $n\to +\infty$. By Lemma
\ref{rel_ent} in the Appendix $\lim_{n\to +\infty}X_n=I(\rho,
\Phi)$. Since
$$
\Phi_n\otimes\Id_R(\hat{\rho})=\Tr_E\, (I_B\otimes P_n\otimes I_R)\cdot (V\otimes
I_R)\cdot \hat{\rho}\cdot (V^*\otimes I_R) \cdot (I_B\otimes P_n\otimes I_R),
$$
where $P_n=\sum_{i=1}^n \ketbra {h_i} {h_i}$ is a finite
dimensional projector in the space $\mathcal{H}_{E}$ and the partial
trace is taking in the space
$\mathcal{H}_B\otimes\mathcal{H}_E\otimes\mathcal{H}_R$, the
operator $\Phi_n\otimes\Id_R(\hat{\rho})$ is isomorphic to the
operator
$$
\widetilde{\Phi}_n(\rho)= \Tr_{BR}\,(I_B\otimes P_n\otimes I_R)\cdot (V\otimes
I_R)\cdot \hat{\rho}\cdot (V^*\otimes I_R)\cdot(I_B\otimes P_n\otimes I_R),
$$
where $\widetilde{\Phi}_n(\cdot)=P_{n}\widetilde{\Phi}(\cdot)P_{n}$
is the quantum operation complementary to the operation $\Phi_n$.
Thus $S(\Phi_n\otimes
\Id_R(\hat{\rho}))=S(\widetilde{\Phi}_n(\rho))$. By using
(\ref{log}) and by noting that $\Phi_n(\cdot)\le\Phi(\cdot)$ we
obtain
\begin{multline*}
-\Tr(\Phi_n \otimes \Id_R (\hat{\rho}))\LOG (\Phi(\rho)\otimes \rho)\\
\begin{aligned}
&= -\Tr(\Phi_n \otimes \Id_R (\hat{\rho}))(\LOG (\Phi(\rho)) \otimes I_R)-\Tr(\Phi_n
\otimes \Id_R (\hat{\rho}))(I_B\otimes\LOG(\rho))\\ &=-\Tr\Phi_n(\rho) \LOG
(\Phi(\rho))-\Tr\left(\Tr_B\Phi_n \otimes \Id_R(\hat{\rho})\right)\LOG(\rho).
\end{aligned}
\end{multline*}

Consider the value
$$
\begin{aligned}
Y_n&=H\bigl(\widetilde {\Phi}_n\otimes \Id_R (\hat{\rho})\|
\widetilde{\Phi}_n(\rho)\otimes\rho\bigr)\\ &=-S\bigl(\widetilde{\Phi}_n\otimes
\Id_R(\hat{\rho})\bigr)-\Tr(\widetilde{\Phi}_n \otimes \Id_R (\hat{\rho}))\LOG
(\widetilde{\Phi}_n(\rho)\otimes \rho).
\end{aligned}
$$
By Lemma 1  $\lim_{n\to +\infty}Y_n=I(\rho,
\widetilde{\Phi})$.  Similar to the calculations of the summands
of $X_n$ we obtain
$$
S\bigl(\widetilde{\Phi}_n \otimes \Id_R(\hat{\rho})\bigr)=S(\Phi_n(\rho))
$$
and
\begin{multline*}
-\Tr\bigl(\widetilde{\Phi}_n \otimes \Id_R
(\hat{\rho})\bigr)\LOG\bigl(\widetilde{\Phi}_n(\rho)\otimes \rho\bigr)\\
\begin{aligned}
&=-\Tr\bigl(\widetilde{\Phi}_n \otimes\Id_R
(\hat{\rho})\bigr)\bigl(\LOG(\widetilde{\Phi}_n(\rho))\otimes I_R\bigr)
-\Tr\bigl(\widetilde{\Phi}_n \otimes \Id_R (\hat{\rho})\bigr)(I_E \otimes \LOG
(\rho))\\ &= S\bigl(\widetilde{\Phi}_n(\rho)\bigr)-\Tr\bigl(\Tr_E \widetilde{\Phi}_n
\otimes \Id_R (\hat{\rho})\bigr)\LOG (\rho).
\end{aligned}
\end{multline*}

Let us show that  $\lim_{n\to +\infty} (X_n+Y_n)=2H(\rho)$. By the
definition of the relative entropy we have
$$
\begin{aligned}
X_n+Y_n&=-\Tr\Phi_n(\rho)\LOG(\Phi(\rho))-S(\Phi_n(\rho))+C_n+D_n+R_n\\ &=
H(\Phi_n(\rho)\| \Phi(\rho))+C_n+D_n,
\end{aligned}
$$
where
$$
C_n=-\Tr\left(\Tr_B \Phi_n \otimes \Id_R(\hat{\rho})\right)\LOG (\rho),\qquad D_n
=-\Tr\bigl(\Tr_E \widetilde{\Phi}_n \otimes \Id_R (\hat{\rho})\bigr)\LOG(\rho).
$$

Let us prove that $\lim_{n\to +\infty} C_n=\lim_{n \to +\infty}
D_n=H(\rho) $. By noting that
$$
\Tr_B \Phi_n \otimes \Id_R(\hat{\rho})= \sum_{i,j=1}^m\! \sqrt{\lambda_i \lambda_j}
\Tr\Phi_n (|e_i\rangle\langle e_j|)|e_i\rangle\langle e_j|,
$$
we obtain
$$
C_n=\sum_{i=1}^m (-\lambda_i\LOG \lambda_i) \Tr\Phi_n (|e_i\rangle\langle e_i|),
$$
and hence $\lim_{n\to +\infty}C_n=H(\rho)$, since
$\lim_{n \to +\infty}\Tr\Phi_n (|e_i\rangle\langle e_i|)=\!1$.
In the similar way one can prove that $\lim_{n\to
+\infty}D_n=H(\rho)$. Lemma \ref{rel_ent} in the Appendix implies
$$
\lim_{n\to +\infty} H(\Phi_n(\rho)\| \Phi(\rho))=0.
$$
Thus we have $\lim_{n\to +\infty}(X_n+Y_n)=2H(\rho)$. Since
$\lim_{n\to +\infty}X_n=I(\rho, \Phi)$ and $\lim_{n\to
+\infty}Y_n=I(\rho,\widetilde{\Phi})$, the assertion of the theorem
is proved for finite rank states. Since the left and the right sides
of relation (\ref{basic1}) are concave lower semicontinuous
nonnegative functions (by Proposition \ref{concave}), validity of
this relation for all states follows from lemma 6 in \cite{Sh-11},
stated that any concave lower semicontinuous lower bounded function
on the set of quantum states is uniquely determined by its
restriction to the set of finite rank states.
\end{proof}

\section{Coherent information}

Since in the infinite dimensional case the right side in definition
(\ref{c-inf}) of the coherent information $I_c(\rho,\Phi)$ may not
be defined even for the state $\rho$ with finite entropy while the
results of Section~4 show finiteness of the mutual information
$I(\rho,\Phi)$ for any such state $\rho$ and any channel $\Phi$, it
seems natural to use the following definition of the coherent
information for an infinite dimensional quantum channel.\medskip

\begin{definition}\label{c-inf+}
Let $\Phi\colon \mathfrak{S}(\mathcal{H}_A) \to
\mathfrak{S}(\mathcal{H}_B)$ be a quantum channel and $\rho$ be a
state in $\mathfrak{S}(\mathcal{H}_A)$ with finite entropy. The
\emph{coherent information} of the channel $\Phi$ at the state
$\rho$ is defined as follows
$$ I_c(\rho,\Phi)=I(\rho,\Phi)-H(\rho).
$$
\end{definition}

In the case $H(\rho)<+\infty$ and $H(\Phi(\rho))<+\infty$ this definition is consistent with the conventional one, since
$I(\rho,\Phi)=H(\rho)+H(\Phi(\rho))-H(\tilde{\Phi}(\rho))$ and hence
\begin{equation}\label{ic}
I_c(\rho,\Phi)=H(\Phi(\rho))-H(\tilde{\Phi}(\rho)).
\end{equation}

The above-defined value inherits properties 2,3 of the mutual
information (see Proposition \ref{concave}). Theorem 1 implies the inequalities
$$
-H(\rho)\le I_c(\rho,\Phi)\le H(\rho)
$$
and a generalization of identity (\ref{main1}) to the infinite
dimensional case. \medskip

\begin{corollary}\label{main_prop+}
\emph{Let $\Phi\colon \mathfrak{S}(\mathcal{H}_A)\to
\mathfrak{S}(\mathcal{H}_B)$ be a quantum channel and
$\widetilde{\Phi}\colon \mathfrak{S}(\mathcal{H}_A) \to
\mathfrak{S}(\mathcal{H}_E)$ be the channel complementary to the
channel $\Phi$. For an arbitrary state
$\rho\in\mathfrak{S}(\mathcal{H}_A)$ with finite entropy the
following relation holds:}
\begin{equation}\label{basic2}
I_c(\rho,\Phi)+I_c(\rho,\widetilde{\Phi})=0.
\end{equation}
\end{corollary}

\begin{remark}\label{c-inf+r}
An alternative expression for the coherent information of the
channel $\Phi$ at the state $\rho$ with finite entropy can be
obtained by using the relation of this quantity with the secret
classical capacity of a channel mentioned in \cite{Sch}. Consider
the $\chi$-function of the channel $\Phi$ defined as follows
$$
\chi_{\Phi}(\rho)=\sup\sum_i\pi_iH(\Phi (\rho_i)\| \Phi(\rho)),\quad
\rho\in\mathfrak{S}(\mathcal{H}),
$$
where the supremum is taken over all convex decompositions
$\rho=\sum_i\pi_i\rho_i$,
$\rho_i\in\mathfrak{S}(\mathcal{H})$. This function is closely
connected to the classical capacity of the channel $\Phi$
(cf.~\cite{QSCI}). If $H(\Phi(\rho))<+\infty$ then
$\chi_{\Phi}(\rho)=H(\Phi(\rho))-\overco H_{\Phi}(\rho)$, where
$\overco H_{\Phi}(\rho)$ is the convex closure of the output entropy
of the channel $\Phi$ (cf.~\cite{H-Sh}). Since $\overco
H_{\Phi}\equiv\overco H_{\widetilde{\Phi}}$ and
$|H(\Phi(\rho))-H(\widetilde{\Phi}(\rho))|\le H(\rho)$ by the
triangle inequality \cite{NC}, for an arbitrary state $\rho$ such
that $H(\rho)<+\infty$ and $H(\Phi(\rho))<+\infty$ we have
\begin{equation}\label{c-inf-new-def}
I_c(\rho,\Phi)=\chi_{\Phi}(\rho)-\chi_{\widetilde{\Phi}}(\rho).
\end{equation}

Since
$\max\{\chi_{\Phi}(\rho),\chi_{\widetilde{\Phi}}(\rho)\}\leq
H(\rho)$ by monotonicity of the relative entropy, the right side in
(\ref{c-inf-new-def}) is a correctly defined value in
$\,[-H(\rho),H(\rho)]\,$ under the single condition $H(\rho)<+\infty$ (for arbitrary value of $H(\Phi(\rho))$), which can be used for definition of
$I_c(\rho,\Phi)$.

\emph{Definition \ref{c-inf+} of the coherent information coincides with the definition
given by (\ref{c-inf-new-def}) for all states with finite entropy.} This assertion is proved in Example \ref{new-e} after the below
Proposition \ref{cont-cond+}. $\square$
\end{remark}\medskip

In finite dimensions the equality $H(\rho)=I_c(\rho,\Phi)$ is a
necessary and sufficient condition of perfect reversibility of the
channel $\Phi$ on the state $\rho$ (see \cite[Theorem 12.10]{NC}).
Generalize this to the infinite dimensional case.\medskip

\begin{definition}\label{dobr} A channel $\Phi$ is called \emph{perfectly reversible on a state}
$\rho$ in $\mathfrak{S}(\mathcal{H}_A)$ if there exists a channel
$\mathcal{D}\colon \mathfrak{S}(\mathcal{H}_B)\to
\mathfrak{S}(\mathcal{H}_A)$ such that
$$
{\mathcal D}\circ\Phi (\tilde{\rho})= \tilde{\rho}
$$
 for all states $\tilde{\rho}$ with supp $\tilde{\rho}\subset\mathcal{L}\equiv$ supp $\rho$.
\end{definition}\medskip

In other words the subspace $\mathcal{L}$ is a quantum code
correcting errors of the channel $\Phi$ \cite{NC}. Introduce the
reference system $\mathcal{H}_R$ and consider a purification
$\rho_{AR}=\ketbra {\varphi_{AR}}{\varphi_{AR}}\in
\mathfrak{S}(\mathcal{H}_A \otimes \mathcal{H}_R)$ of the state
$\rho$.\medskip

\begin{lemma}\label{tobr} \emph{A channel $\Phi$ is perfectly reversible on a state
$\rho\in\mathfrak{S}(\mathcal{H}_A)$ if and only if there exists a
channel $\mathcal{D}\colon \mathfrak{S}(\mathcal{H}_B)\to
\mathfrak{S}(\mathcal{H}_A)$ such that}
\begin{equation}\label{eobr}
(\mathcal{D} \circ \Phi \otimes \Id_R)(\rho_{AR})=\rho_{AR}.
\end{equation}
\end{lemma}

The proof of this lemma is presented in the Appendix.\medskip

\begin{property}
\emph{Let $H(\rho)<\infty$. A channel $\Phi\colon
\mathfrak{S}(\mathcal{H}_A)\to \mathfrak{S}(\mathcal{H}_B )$ is
perfectly reversible on the state $\rho$ if and only if one of
following equivalent conditions holds: $ I_c(\rho, \Phi) =
H(\rho);\; I(\rho, \widetilde{\Phi})=0.$}
\end{property}

\begin{proof} By Theorem 1
and Definition \ref{c-inf+} we have
$$
H(\rho)-I_c(\rho,\Phi)=I(\rho, \widetilde {\Phi})\ge 0,
$$
where the equality  holds if and only if
$\rho_{RE} = \rho_R \otimes \rho_E$, since $I(\rho, \widetilde
{\Phi})= H(\rho_{RE}||\rho_R \otimes \rho_E)$. The following part of
the proof is similar to the proof presented in \cite{QSCI} and we
give it here for completeness.

\underline{Necessity}. Let $V\colon \mathcal{H}_A\to \mathcal{H}_B
\otimes \mathcal{H}_E$ be the isometry from representation
(\ref{compl}) of the channel $\Phi$. Consider the pure state
$\rho_{BRE}=\ketbra {\varphi_{BRE}} {\varphi_{BRE}}$, where $\ket
{\varphi_{BRE}}=(V \otimes I_R) \ket {\varphi_{AR}}$. Since the
channel $\Phi$ is perfectly reversible, there exists a channel
$\mathcal{D}$ such that (\ref{eobr}) holds and hence
$$
(\mathcal{D} \otimes \Id_{RE})(\rho_{BRE})=\rho_{ARE}.
$$
Since $\rho_{AR}$ is a pure state, we have $\rho_{ARE} = \rho_{AR}
\otimes \rho_E$. By taking partial traces over the space
$\mathcal{H}_A$, we obtain $\rho_{RE} = \rho_R \otimes \rho_E$.

\underline{Sufficiency}.  Consider the vector $\ket
{\varphi_{BRE}}=(V \otimes I_R) \ket {\varphi_{AR}}$. Then $\ket
{\varphi_{BRE}}$ is a purification vector for the state $\rho_{RE}$.
Since $\rho_{RE}=\rho_R \otimes \rho_E$, $\ket {\varphi_{AR}}
\otimes \ket {\varphi_{EE'}}$ is a purification vector for the state
$\rho_{RE}$, where $E'$ is a reference system for the system $E$.

Without loss of generality we can assume that the Hilbert spaces of
the both purifications are infinite dimensional, so that there
exists an isometry $W\colon \mathcal{H}_B\to \mathcal{H}_A \otimes
\mathcal{H}_{E'}$ such that
$$
(I_{RE} \otimes W)\ket {\varphi_{BRE}}=\ket {\varphi_{AR}} \otimes \ket
{\varphi_{EE'}},
$$
and respectively
$$
(I_{RE} \otimes W)\ket {\varphi_{BRE}} \bra {\varphi_{BRE}} (I_{RE} \otimes
W^*)=\ketbra {\varphi_{AR}} {\varphi_{AR}} \otimes \ketbra {\varphi_{EE'}}
{\varphi_{EE'}}.
$$
By taking partial traces over the spaces $\mathcal{H}_E$ and
$\mathcal{H}_{E'}$, we obtain perfect reversibility condition
(\ref{eobr}), where
$$
\mathcal{D}(\sigma)=\Tr_{E'} W \sigma W^*,\quad \sigma\in \mathfrak{S}
(\mathcal{H}_B).
$$
\end{proof}

As mentioned before, the entropy $H(\rho)$ of a state $\rho$ is
the upper bound for the  coherent information $I_c(\rho,\Phi)$ of an
arbitrary channel $\Phi$ at this state. The following proposition
gives the more precise upper bound for $I_c(\rho,\Phi)$, expressed
via the Kraus operators of the channel $\Phi$.\medskip

\begin{property}\label{a-a-e-lemma} \emph{Let $\Phi(\cdot)=\sum_{i=1}^{+\infty} V_i(\cdot)V_i^*$
be a quantum channel. Then for an arbitrary state $\rho$ with finite
entropy the following inequality holds
\begin{equation}\label{a-a-e}
I_c(\rho,\Phi)\le\sum_{i=1}^{+\infty} H(V_i\rho V_i^{*})=\sum_{i=1}^{+\infty} \Tr V_i\rho V_i^{*}H\left(\frac{V_i\rho V_i^{*}}{\Tr V_i\rho V_i^{*}}\right).
\end{equation}
The equality holds in (\ref{a-a-e})  if  $\,\Ran V_i\perp\Ran V_j\,$ for
all $i\ne j$.}
\end{property}\medskip

The expression in the right side of (\ref{a-a-e}) can be considered
as the mean entropy of a posteriori state in quantum measurement,
described by the collection of operators
$\left\{V_i\right\}_{i=1}^{+\infty}$, at a priory  state $\rho$
\cite{QSCI,NC}. By the Groenevold-Lindblad-Ozawa
inequality this value does not exceed $H(\rho)$ \cite{O}.

\begin{proof} Show first that the equality holds in
(\ref{a-a-e}) if  $\mathrm{Ran}V_{i}\perp\mathrm{Ran}V_{j}$ for all
$i\ne j$.

Let  $\rho$ be a state in $\mathfrak{S}(\mathcal{H}_A)$ with finite
entropy and $\ket{\varphi}\bra {\varphi}$ be its purification in
$\mathfrak{S}(\mathcal{H}_A\otimes\mathcal{H}_R)$.  By using the
well known properties of the relative entropy (see \cite{L}) and by
noting that $\sum_{i=1}^{+\infty} V^*_iV_i=I_A$ we obtain
$$
\begin{aligned}
I(\rho,\Phi)&=H\left(\Phi \otimes \Id_R (\ket {\varphi} \bra {\varphi}) \| \, \Phi
\otimes \Id_R (\rho\otimes\rho)\right)\\ &=H\left(\sum_{i=1}^{+\infty} V_i\otimes I_R
\ket {\varphi} \bra {\varphi} V^*_i\otimes I_R\mathrel{\Big\|}
\sum_{i=1}^{+\infty}(V_i\otimes I_R) (\rho\otimes \rho)(V^*_i\otimes I_R)\right)\\
&=\sum_{i=1}^{+\infty} H\left(V_i\otimes I_R \ket {\varphi} \bra {\varphi} V^*_i\otimes
I_R\| (V_i\otimes I_R) (\rho\otimes \rho) (V^*_i\otimes I_R)\right)\\
&=H(\rho)+\sum_{i=1}^{+\infty}\left[S(V_i\rho V_i^*)-\eta(\Tr V_i\rho
V_i^*)\right]=H(\rho)+\sum_{i=1}^{+\infty} H(V_i\rho V_i^*).
\end{aligned}
$$

Let $\Phi(\cdot)=\sum_{i=1}^{+\infty} V_i(\cdot)V_i^*$ be an
arbitrary channel and
$\mathcal{H}_C=\bigoplus_{i=1}^{+\infty}\mathcal{H}^i_B$, where
$\mathcal{H}^i_B\cong\mathcal{H}_B$. Let $U_i$ be an isometrical
embedding of $\mathcal{H}_B$ in $\mathcal{H}_C$ such that
$U_i\mathcal{H}_B=\mathcal{H}^i_B$ for each $i$.

As proved before, for the quantum channel
$\widehat{\Phi}(\cdot)=\sum_{i=1}^{+\infty}
U_iV_i(\cdot)V_i^*U_i^*$ the following equality holds
$$
I_c(\rho,\widehat{\Phi})=\sum_{i=1}^{+\infty} H(U_iV_i\rho V_i^*U_i^*)=\sum_{i=1}^{+\infty}
H(V_i\rho V_i^*),
$$
By applying the 1-th chain rule for the coherent information to the
composition $\Psi\circ\widehat{\Phi}=\Phi$, where
$\Psi(\cdot)=\sum_{i=1}^{+\infty} U_i^*(\cdot)U_i$ is a channel
from $\mathfrak{S}(\mathcal{H}_C)$ to $\mathfrak{S}(\mathcal{H}_B)$,
we obtain (\ref{a-a-e}) from the above equality.
\end{proof}

\section{On continuity of mutual and coherent informations}

Proposition \ref{concave} and Theorem 1 provide the following
continuity condition for mutual  and  coherent informations.\medskip

\begin{property}\label{cont-cond}
\emph{For an arbitrary quantum channel $\,\Phi\colon
\mathfrak{S}(\mathcal{H}_A)\to \mathfrak{S}(\mathcal{H}_B)$ the
functions
$$
\rho\mapsto I(\rho,\Phi)\qquad \text{and\/}\qquad\rho\mapsto I_c(\rho,\Phi)
$$
are continuous on any subset
$\mathcal{A}\subset\mathfrak{S}(\mathcal{H}_A)$, on which the von
Neumann entropy is continuous.}
\end{property}

\begin{proof}
By Proposition \ref{concave} the functions $\rho\mapsto
I(\rho,\Phi)$ and
 $\rho\mapsto I(\rho,\widetilde{\Phi})$ are lower semicontinuous while
by Theorem 1 their sum coincides with the double von Neumann
entropy, which is continuous on the set $\mathcal{A}$ by the
condition. Hence these functions are continuous on the set
$\mathcal{A}$. The function $\rho\mapsto I_c(\rho,\Phi)$ is
continuous on the set $\mathcal{A}$ as a difference between two
functions continuous on this set.
\end{proof}

\begin{example}\label{cont-cond-e}
Let $H$ be a Hamiltonian of quantum system $A$. Then the subset
$\mathcal{K}_{H,h}$ of $\mathfrak{S}(\mathcal{H}_A)$, consisting of
states $\rho$ such that $\Tr H\rho\le h$, can be treated as a set of
states with the mean energy not exceeding $h$. If the operator $H$
is such that $\Tr e^{-\lambda H}<+\infty$ for all $\lambda>0$ then
the von Neumann entropy is continuous on $\mathcal{K}_{H,h}$
\cite{O&P,W}. This holds, for example, for the Hamiltonian of the
system of quantum oscillators \cite{W}. By Proposition
\ref{cont-cond} for an arbitrary quantum channel $\Phi$ the mutual
information $I(\rho,\Phi)$ and the coherent information
$I_c(\rho,\Phi)$ are continuous functions of a state
$\rho\in\mathcal{K}_{H,h}$ for each finite $h>0$. Hence these
functions are bounded and achieve their supremum on the set
$\mathcal{K}_{H,h}$ (by compactness of this set \cite{H-Sh-2}).
\end{example}\medskip

\begin{remark}\label{cont-cond-r}
Proposition \ref{cont-cond} and identity (\ref{basic1}) show that
for an arbitrary channel $\Phi$ the function $\rho\mapsto
I(\rho,\Phi)$ is continuous and bounded on the set
$$
\mathfrak{S}_k(\mathcal{H}_A)
=\{\rho\in\mathfrak{S}(\mathcal{H}_A)\mid\rank\rho\le k\}
$$
for each $k=1,2,...$ Hence the properties of the function
$\rho\mapsto I(\rho,\Phi)$ can be explored by using the
approximation method considered in \cite[Section 4]{Sh-11}. This
method makes it possible to clarify the sense of the continuity
condition for the function $\rho\mapsto I(\rho,\Phi)$ in Proposition
\ref{cont-cond} and to show its necessity for the particular class
of channels.

By Proposition 3 in \cite{Sh-11} the function $\rho\mapsto
I(\rho,\Phi)$ is a pointwise limit of the increasing sequence of
concave continuous on the set  $\mathfrak{S}(\mathcal{H}_A)$
functions
\begin{equation}\label{approximator} \rho\mapsto
I_k(\rho,\Phi)=\sup_{\{\pi_i,\rho_i\}\in\mathcal{E}^k_\rho}\sum_i\pi_iI(\rho_i,\Phi),
\end{equation}
where $\mathcal{E}^k_\rho$ is the set of all ensembles\footnote{An ensemble
$\{\pi_{i},\rho_{i}\}$ is a collection of states $\{\rho_{i}\}$ with
the corresponding probability distribution $\{\pi_{i}\}$.}
$\{\pi_{i},\rho_{i}\}$ such that $\sum_i\pi_i\rho_i=\rho$ and
$\rank\rho_i\le k$ for all $i$. Hence a
necessary and sufficient condition of continuity of the function
$\rho\mapsto I(\rho,\Phi)$ on a set
$\mathcal{A}\subset\mathfrak{S}(\mathcal{H}_A)$ can be expressed as
follows
\begin{equation}\label{I-cont-cond}
\lim_{k\to+\infty}\sup_{\rho\in\mathcal{A}_c}\Delta_k^{I}(\rho,\Phi)=0\quad
\text{for any compact set}\; \mathcal{A}_c\subseteq\mathcal{A},
\end{equation}
where $\Delta_k^{I}(\rho,\Phi)=I(\rho,\Phi)-I_k(\rho,\Phi)$. It is
possible to show that
$$
\Delta_k^{I}(\rho,\Phi)
=\inf_{\{\pi_i,\rho_i\}\in\mathcal{E}^k_\rho}\sum_i\pi_i\left[H(\rho_i\|
\rho)+H(\Phi(\rho_i)\| \Phi(\rho)) -H(\widetilde{\Phi}(\rho_i)\|
\widetilde{\Phi}(\rho))\right]
$$
for any state $\rho$ with finite entropy. By monotonicity and
nonnegativity of the relative entropy the expression in the square
brackets does not exceed $2H(\rho_i\| \rho)$. Thus
(\ref{I-cont-cond}) holds if
$$
\lim_{k\to+\infty}\sup_{\rho\in\mathcal{A}_c}\inf_{\{\pi_i,\rho_i\}
\in\mathcal{E}^k_\rho}\sum_i\pi_iH(\rho_i\| \rho)=0\quad \text{for any compact set}\;
\mathcal{A}_c\subseteq\mathcal{A},
$$
which is equivalent to continuity of the entropy on the set
$\mathcal{A}$, since it means uniform convergence of the  sequence
$\{H_{k}\}$ of continuous approximators of the entropy (defined by
formula (\ref{approximator}) with $H(\rho_i)$ instead of
$I(\rho_i,\Phi)$ in the right side) on compact subsets of
$\mathcal{A}$ \cite{Sh-11}. \smallskip

Thus  \emph{the assertion of Proposition \ref{cont-cond} is
explained by the implication
\begin{equation}\label{implication}
H_k(\rho)\mathop{\rightrightarrows}^{\mathcal{A}}
H(\rho)<+\infty\;\,\Rightarrow\;\,
I_k(\rho,\Phi)\mathop{\rightrightarrows}^{\mathcal{A}}
I(\rho,\Phi)<+\infty\;\;\; \forall\mathcal{A}\subset\mathfrak{S}(\mathcal{H}_A)\!
\end{equation}
valid for any channel $\Phi$ by monotonicity of the relative
entropy.}\smallskip

If $\Phi$ is a degradable channel, that is
$\widetilde{\Phi}=\Lambda\circ\Phi$ for some channel $\Lambda$, then
$I(\rho,\Phi)<+\infty\Rightarrow H(\rho)<+\infty$ by Theorem 1 and
the 1-th chain rule from Proposition \ref{concave}, while the
expression in the square brackets in the above formula for
$\Delta_{k}^{I}(\rho,\Phi)$ is not less than $H(\rho_i\|
\rho)$ by monotonicity of the relative entropy. Thus for degradable
channel $\Phi$ $"\Leftrightarrow"$ holds in  (\ref{implication}) and
hence the continuity condition for the function $\rho\mapsto
I(\rho,\Phi)$ in Proposition \ref{cont-cond} is necessary and
sufficient:
$$
\lim_{n\to+\infty} H(\rho_n)=H(\rho_0)<+\infty\quad\Longleftrightarrow\quad
\lim_{n\to+\infty} I(\rho_n,\Phi)=I(\rho_0,\Phi)<+\infty
$$
for any  sequence $\{\rho_n\}$ converging to a state $\rho_0$.
\end{remark}

To explore continuity of capacities as functions of a channel it is
necessary to consider the corresponding entropic characteristics as
functions of a pair (input state, channel), that is as functions on
the Cartesian product  of the set of all input states
$\mathfrak{S}(\mathcal{H}_A)$ and the set of all channels
$\mathfrak{F}(A,B)$ from  $A$ to $B$ endowed with the appropriate
(sufficiently weak) topology. As shown in \cite{H-Sh}, for this
purpose it is reasonable to use the strong convergence topology on
the set $\mathfrak{F}(A,B)$ described before Proposition
\ref{concave} in Section~3. By using the obvious modification of the
arguments in the proof of Proposition \ref{cont-cond} one can derive
from Proposition \ref{concave} and Theorem 1 the following result. \medskip
\begin{property}\label{cont-cond+}
\emph{Let $\,\{\Phi_{n}\}$ be a sequence of channels in
$\,\mathfrak{F}(A,B)$ strongly converging to a channel $\,\Phi_{0}$
and there exists a sequence $\,\{\widetilde{\Phi}_{n}\}$ of channels
in $\,\mathfrak{F}(A,E)$ strongly converging to a channel
$\,\widetilde{\Phi}_{0}$ such that $(\Phi_{n},\widetilde{\Phi}_{n})$
is a complementary pair for each $n=0,1,2,\ldots\strut$. Then the
relations
\begin{equation}\label{lim-exp}
\lim_{n\to+\infty}I(\rho_n,\Phi_n)=I(\rho_0,\Phi_0)\qquad\text{and\/}\qquad
\lim_{n\to+\infty}I_c(\rho_n,\Phi_n)=I_c(\rho_0,\Phi_0)
\end{equation}
hold for any sequence  $\,\{\rho_{n}\}$ of states in
$\mathfrak{S}(\mathcal{H}_A)$ converging to the state $\rho_{0}$
such that $\,\lim_{n\to+\infty}H(\rho_n)=H(\rho_0)<+\infty$.}
\end{property}\medskip

\begin{example}\label{new-e}
By using Proposition \ref{cont-cond+} one can prove representation (\ref{c-inf-new-def}) for any state $\rho$ with finite entropy as follows.
Let $\Phi(\cdot)=\sum_{i=1}^{+\infty}V_{i}(\cdot)V^{*}_{i}$. Consider the sequence
of channels
$\Phi_n(\cdot)=\sum_{i=1}^{n}V_{i}(\cdot)V^{*}_{i}+W_n(\cdot)W_n$,
where $W_n=\sqrt{I_A-\sum_{i=1}^{n}V^{*}_{i}V_{i}}$. By noting that the sequence $\{W_n\}$ converges to the zero operator in the strong operator topology it is easy to show that the sequences $\{\Phi_n\}$ and $\{\widetilde{\Phi}_n\}$ strongly converges to the channels $\Phi$ and $\widetilde{\Phi}$.
Hence Proposition \ref{cont-cond+} implies
$$
\lim_{n\to+\infty}I_c(\rho,\Phi_n)=I_c(\rho,\Phi).
$$
Let $\Psi_n(\cdot)=\sum_{i=1}^{n}V_{i}(\cdot)V^{*}_{i}$ and $\Theta_n(\cdot)=W_n(\cdot)W_n$ be quantum operations.
By Proposition 6B in \cite{H-Sh} we have $\,\lim_{n\to+\infty}\chi_{\Psi_n}(\rho)=\chi_{\Phi}(\rho)$ while Corollary 8 in \cite{Sh-11} implies $\lim_{n\to+\infty}\chi_{\Theta_n}(\rho)=\lim_{n\to+\infty}H(W_n\rho W_n)=0$. By using Corollary 3 in \cite{H-Sh} we conclude that
$\lim_{n\to+\infty}\chi_{\Phi_n}(\rho)=\chi_{\Phi}(\rho)$. Since $H(\rho)<+\infty$ implies
$H(\Phi_n(\rho))<+\infty$ for all $n$ we have (see Remark \ref{c-inf+r})
$$
I_c(\rho,\Phi_n)=\chi_{\Phi_n}(\rho)-\chi_{\widetilde{\Phi}_n}(\rho).
$$
By the above observations and Corollary 3 in \cite{H-Sh} passing to the limit in this equality leads to the inequality
$$
I_c(\rho,\Phi)\leq\chi_{\Phi}(\rho)-\chi_{\widetilde{\Phi}}(\rho).
$$
By using  the same approximation for the channel $\widetilde{\Phi}$ instead of $\Phi$ and repeating the above arguments we obtain the converse inequality. $\square$
\end{example}\medskip

Let $\mathfrak{V}_{1}(A,B)$ be the set of all sequences
$\overline{V}=\left\{V_i\right\}_{i=1}^{+\infty}$ of operators from
$\mathcal{H}_A$ into $\mathcal{H}_B$ such that
$\sum_{i=1}^{+\infty}V^*_iV_i=I_A$ endowed with the topology
of coordinate\nobreakdash-\hspace{0pt}wise strong operator
convergence.\medskip
\begin{corollary}\label{cont-cond+c}
\emph{For an arbitrary subset
$\mathcal{A}\subset\mathfrak{S}(\mathcal{H}_A)$, on which the von
Neumann entropy is continuous, the functions
$$
\bigl(\rho, \overline{V}\bigr)\mapsto I\bigl(\rho,\Phi[\overline{V}]\bigr),\qquad
\bigl(\rho, \overline{V}\bigr)\mapsto I_c\bigl(\rho,\Phi[\overline{V}]\bigr),\qquad
\bigl(\rho, \overline{V}\bigr)\mapsto \sum_{i=1}^{+\infty}H(V_i\rho V_i^*),
$$
where
$\Phi[\overline{V}](\cdot)=\sum_{i=1}^{+\infty}V_i(\cdot)V_i^*$\textup,
are continuous on the set $\mathcal{A}\times\mathfrak{V}_1(A,B)$.
}\end{corollary}

\begin{proof} By Proposition \ref{cont-cond+} continuity of the first two functions follows from continuity of the maps
$$
\mathfrak{V}_1(A,B)\ni\overline{V}\mapsto
\Phi[\overline{V}]\in\mathfrak{F}(A,B)\quad \textup{and}\quad
\mathfrak{V}_1(A,B)\ni\overline{V}\mapsto
\widetilde{\Phi}[\overline{V}]\in\mathfrak{F}(A,E),
$$
where
$\widetilde{\Phi}[\overline{V}](\cdot)=\sum_{i,j=1}^{+\infty}[\Tr
V_i(\cdot)V_j^*]|h_i\rangle\langle h_j|$ and $\{|h_i\rangle\}$ is a
particular orthonormal basis in $\mathcal{H}_E$.

To prove continuity of the above maps it suffices to show that
\begin{equation}\label{maps-1}
\lim_{n\to+\infty}\Phi[\overline{V}_n](|\varphi\rangle\langle\varphi|)
=\Phi[\overline{V}_0](|\varphi\rangle\langle\varphi|)
\end{equation}
and
\begin{equation}\label{maps-2}
\lim_{n\to+\infty}\widetilde{\Phi}[\overline{V}_n](|\varphi\rangle\langle\varphi|)
=\widetilde{\Phi}[\overline{V}_0](|\varphi\rangle\langle\varphi|)
\end{equation}
for any sequence $\{\overline{V}_{n}\}\subset\mathfrak{V}_{1}(A,B)$,
converging to a vector $\overline{V}_0\in\mathfrak{V}_1(A,B)$, and
for any unit vector $\varphi\in\mathcal{H}_A$.

Let $\overline{V}_n=\left\{V^n_i\right\}_{i=1}^{+\infty}$ for each
$n\ge0$. Relation (\ref{maps-1}) can be proved by noting that the
condition $\sum_{i=1}^{+\infty}\|V^n_i|\varphi\rangle\|^2=1$
for all $n\ge0$ implies
$$
\lim_{m\to+\infty}\sup_{n\ge0}\Tr\sum_{i>m}V^n_i\,|\varphi\rangle\langle\varphi|(V^n_i)^*
=\lim_{m\to+\infty}\sup_{n\ge0}\sum_{i>m}\|V^n_i|\varphi\rangle\|^2=0.
$$
Relation (\ref{maps-2}) is easily proved by using the result from
\cite{d}, mentioned before Proposition  \ref{concave}.

To prove continuity of the third  function consider the following
construction. Let
$\mathcal{H}_C=\bigoplus_{i=1}^{+\infty}\mathcal{H}^i_B$,
where $\mathcal{H}^i_B\cong\mathcal{H}_B$, and $U_i$ be an
isometrical embedding of $\mathcal{H}_B$ into $\mathcal{H}_C$ such
that $U_{i}\mathcal{H}_B=\mathcal{H}^{i}_B$ for each $i$.

For an arbitrary sequence $\left\{V_i\right\}_{i=1}^{+\infty}$ in
$\mathfrak{V}_1(A,B)$ one can take the sequence
$\left\{\hat{V}_i=U_iV_i\right\}_{i=1}^{+\infty}$ in
$\mathfrak{V}_1(A,C)$ such that $\Ran\hat{V}_i\perp\Ran\hat{V}_j$
for all $i\ne j$. Since the above correspondence is continuous (as a
map from $\mathfrak{V}_{1}(A,B)$ into $\mathfrak{V}_{1}(A,C)$) the
above observation shows continuity on the set
$\mathcal{A}\times\mathfrak{V}_{1}(A,B)$ of the function
$$
\bigl(\rho,\overline{V}\bigr)\mapsto
I_c\bigl(\rho,\widehat{\Phi}[\overline{V}]\bigr)=\sum_{i=1}^{+\infty}
H\bigl(\hat{V}_i\rho \hat{V}_i^*\bigr)=\sum_{i=1}^{+\infty}H(V_i\rho V_i^*),
$$
where $\widehat{\Phi}[\overline{V}](\cdot)
=\sum_{i=1}^{+\infty}\hat{V}_i(\cdot)\hat{V}_i^*$ and the
first equality follows from the last assertion of Proposition
\ref{a-a-e-lemma}.
\end{proof}

As mentioned in Section~5, the value
$\sum_{i=1}^{+\infty}H(V_i\rho V_i^*)$ can be considered as
the mean entropy of a posteriori state in the quantum measurement,
described by the collection of operators
$\left\{V_i\right\}_{i=1}^{+\infty}$. Corollary \ref{cont-cond+c}
shows that continuity of the entropy $H(\rho)$ of a priory state
$\rho$ implies continuity of the mean entropy of a posteriori state
as a function of a pair (a priori state, measurement) provided the
strong operator topology is used in the definition of convergence of
a sequence of measurements. This assertion strengthens the analogous
assertion in Example 3 in \cite{Sh-11}, in which  the
\emph{stronger} topology (so called the
$*$\nobreakdash-\hspace{0pt}strong operator topology) is used in
definition of convergence of a sequence of
measurements.\footnote{Note that this stronger version can not be
proved by means of the method used in \cite{Sh-11}.} Hence by means
of Corollary \ref{cont-cond+c} one can strengthen all the assertions
in Example 3 in \cite{Sh-11} by inserting the strong operator
topology in definition of convergence of a sequence of measurements,
which seems more natural in this context. \medskip

\section{Appendix}

\begin{lemma}\label{ohya_petz}
\emph{Let $\rho$ and $\sigma$ be states in $\mathfrak{S}(\mathcal{H})$ and
$C$ be an operator in $\,\mathfrak{T}_{+}(\mathcal{H})$. Then
$$
H(\lambda \rho+(1-\lambda)\sigma\| C)\ge \lambda H(\rho\| C)+(1-\lambda)
H(\sigma\| C)-h_2(\lambda),\;\; \forall\lambda\in [0,1],
$$
where $h_2(\lambda)=\eta(\lambda)+\eta(1-\lambda)$.}
\end{lemma}

\begin{proof} Let $\{P_n\}$ be an increasing
sequence of finite rank projectors strongly converging to the
identity operator. Then $A_n = P_n \rho P_n$, $B_n = P_n \sigma P_n$
and $C_n = P_n C P_n$ are finite rank operators for each  $n$ and
hence
\begin{multline*}
H(\lambda A_n+(1-\lambda) B_n \| C_n)=\Tr (\lambda A_n+(1-\lambda)B_n) (-\LOG
C_n)\\
\begin{aligned}
&-S(\lambda A_n+(1-\lambda)B_n)+\Tr C_n-\Tr(\lambda A_n+(1-\lambda)B_n)\\ &\ge
\lambda\Tr A_n (-\LOG C_n)+(1-\lambda)\Tr B_n (-\LOG C_n)+\Tr C_n-\lambda\Tr
A_n-(1-\lambda)\Tr B_n\\ &-\lambda S(A_n)-(1-\lambda)S(B_n)-\eta(\Tr(\lambda
A_n+(1-\lambda)B_n))+\lambda\eta(\Tr A_n)\\ & +(1-\lambda)\eta(\Tr B_n)-x_n
h_2(x_n^{-1}\lambda\Tr A_n)=\lambda H( A_n\| C_n)+(1-\lambda)H(B_n\| C_n)\\
&-\eta(\Tr(\lambda A_n+(1-\lambda)B_n))+\lambda\eta(\Tr A_n)+(1-\lambda)\eta(\Tr
B_n)- x_n h_2(x_n^{-1}\lambda\Tr A_n),
\end{aligned}
\end{multline*}
where $x_n=\Tr(\lambda A_n+(1-\lambda)B_n)$ and the inequality
$$
H(\lambda A_n+(1-\lambda)B_n)\le \lambda H(A_n)+(1-\lambda)H(B_n)+x_n
h_2(x_n^{-1}\lambda\Tr A_n),
$$
following from (\ref{H-fun-eq}) and (\ref{H-fun-ineq}) was used. By
Lemma 1 passing to the limit $n\to+\infty$ implies the desired
inequality.
\end{proof}

\begin{lemma}\label{rel_ent}
\emph{Let $\{ A_n \}$ be a sequence of operators in
$\mathfrak{T}_{+}(\mathcal{H})$ converging in the trace norm to an
operator $A_0$ such that  $A_n \le A_0$ for all $n$. Then
$$
\lim_{n\to +\infty} H(A_n\| B)=H(A_0\| B)\quad \text{for any
 operator}\; B\in
\mathfrak{T}_+(\mathcal{H}).
$$}
\end{lemma}

\begin{proof}
We can assume that $A_0$ is a state. It can be represented as
follows
$$
A_0=\lambda_n \rho_n+(1-\lambda_n) \sigma_n,
$$
where
$$
\lambda_n=\Tr A_n,\qquad \rho_n=\frac {A_n} {\Tr A_n},\qquad \sigma_n=\frac {A-A_n}
{1-\lambda_n}.
$$
By Lemma \ref{ohya_petz} and nonnegativity of the relative entropy
we have
$$
\begin{aligned}
H(A_0\| B)&\ge \lambda_n H(\rho_n\| B)+(1-\lambda_n) H(\sigma_n\|
B)-h_2(\lambda_n)\\ &\ge H(A_n\| \lambda_n B)- h_2(\lambda_n)\\ &=H(A_n\|
B)-\Tr B(1-\lambda_n)-\lambda_n\LOG(\lambda_n)- h_2(\lambda_n),
\end{aligned}
$$
and hence  $\limsup_{n\to +\infty} H(A_n\| B)\le H(A_0\|
B)$. By lower semicontinuity of the  relative entropy this implies
the assertion of the lemma.
\end{proof}

\noindent\textit{Proof of Lemma~\ref{tobr}.} Let $T=\mathcal{D}\circ\Phi$. Consider
the set of conditions
\begin{align}
T(\ketbra{\psi}{\psi})&= \ketbra{\psi}{\psi},\quad \forall \ket{\psi}\in
\supp\rho,\label{A1}\\ T(\ketbra{\psi}{\phi})&= \ketbra{\psi}{\phi},\quad \forall
\ket{\psi},\ket{\phi}\in \supp\rho,\label{A2}\\ T(\ketbra{e_i}{e_j})&=
\ketbra{e_i}{e_j},\quad \forall i,j,\label{A3}
\end{align}
where $\ket{e_i}$ is the set of eigenvectors of the state $\rho$
corresponding to nonzero eigenvalues. Then Definition \ref{dobr}
$\Leftrightarrow$ (\ref{A1}) follows from the spectral
representation, (\ref{A1}) $\Leftrightarrow$ (\ref{A2}) follows from
the polarization identity, (\ref{A2}) $\Leftrightarrow$ (\ref{A3})
is obvious, (\ref{A3}) $\Leftrightarrow$ (\ref{eobr}) follows from
formula (\ref{purif})).

\bigskip
The authors are grateful to the participants of the seminar
''Quantum probability, statistics, information'' (MIAN) for the
useful discussion. The authors are also grateful to A.A. Kuznetsova
for the discussion and the help in preparing the manuscript.

\end{document}